%% file: main.tex
\documentclass[sigconf]{acmart}
\copyrightyear{2021} 
\acmYear{2021} 
\setcopyright{acmlicensed}
\acmConference[PODC '21]{Proceedings of the 2021 ACM Symposium on Principles of Distributed Computing}{July 26--30, 2021}{Virtual Event, Italy}
\acmBooktitle{Proceedings of the 2021 ACM Symposium on Principles of Distributed Computing (PODC '21), July 26--30, 2021, Virtual Event, Italy}
\acmPrice{15.00}
\acmDOI{10.1145/3465084.3467904}
\acmISBN{978-1-4503-8548-0/21/07}

\input{defs}

\begin{document}
\fancyhead{}

\title{Hedging Against Sore Loser Attacks in Cross-Chain Transactions}


\author{Yingjie Xue}
\affiliation{%
  \institution{Brown University}
  \city{Providence}
   \state{RI}
  \country{USA}}
\email{yingjie_xue@brown.edu}

\author{Maurice Herlihy}
\affiliation{%
  \institution{Brown University}
  \city{Providence}
   \state{RI}
  \country{USA}}
\email{maurice.herlihy@gmail.com}

\begin{abstract}
\input{abstract}
\end{abstract}


\maketitle

\section{Introduction}
\input{intro}

\section{Related Work}
\input{related}

\section{Model}
\input{model}

\section{Overview}
\input{overview}

\section{Two-Party Swap}
\seclabel{twoparty}
\input{twoparty}

\section{Bootstrapping Premiums}
\seclabel{bootstrap}
\input{bootstrap}

\section{Multi-Party Swap}
\seclabel{multiparty}
\input{multiparty}

\section{Brokered Commerce}
\seclabel{broker}
\input{broker}

\section{Auctions}
\seclabel{auction}
\input{auction}

\section{Remarks and Conclusions}
\seclabel{conclusions}
\input{conclusions}

\begin{acks}
This research was supported by NSF grant 1917990.
\end{acks}

\bibliographystyle{ACM-Reference-Format}
\bibliography{references}

\end{document}

%% file: defs.tex
\usepackage{listings}
\usepackage{subcaption}
\newcommand{\var}[1]{\text{\lstinline+#1+}}

\newtheorem{definition}{Definition}

      \def\cG{\ensuremath{{\mathcal G}}}

\newcommand{\set}[1]{\left\{ #1 \right\}}

\newcommand{\diam}{\ensuremath{\mathit{diam}}}

\usepackage{epsfig}

\newcommand{\figlabel}[1]{\label{figure:#1}}
\newcommand{\nakedfigref}[1]{\ref{figure:#1}}
\newcommand{\figref}[1]{Figure~\nakedfigref{#1}}

\newcommand{\eqnlabel}[1]{\label{eq:#1}}
\newcommand{\nakedeqnref}[1]{\ref{eq:#1}}
\newcommand{\eqnref}[1]{Equation~\nakedeqnref{#1}}

\newcommand{\lemmalabel}[1]{\label{lemma:#1}}
\newcommand{\nakedlemmaref}[1]{\ref{lemma:#1}}
\newcommand{\lemmaref}[1]{Lemma~\nakedlemmaref{#1}}

\newcommand{\lemmarefrange}[2]{Lemmas~\nakedlemmaref{#1}-\nakedlemmaref{#2}}

\newcommand{\seclabel}[1]{\label{sec:#1}}
\newcommand{\nakedsecref}[1]{\ref{sec:#1}}
\newcommand{\secref}[1]{Section~\nakedsecref{#1}}

\usepackage{algorithm}
\usepackage{algorithmic}
\usepackage{xltabular} 
\graphicspath{{./pictures/}}

%% file: abstract.tex
A \emph{sore loser} attack in cross-blockchain commerce rises when one party decides to halt participation partway through, leaving other parties' assets locked up for a long duration. Although vulnerability to sore loser attacks cannot be entirely eliminated, it can be reduced to an arbitrarily low level. This paper proposes new distributed protocols for hedging a range of cross-chain transactions in a synchronous communication model,  such as two-party swaps, $n$-party swaps, brokered transactions, and auctions. 

%% file: intro.tex
Alice is heavily invested in ``apricot tokens'',
an electronic asset managed on the ``apricot blockchain'',
a tamper-proof replicated ledger.
Token prices are volatile, so she decides to diversify.
She locates Bob, who owns "banana tokens",
managed on a distinct ``banana blockchain'',
and Alice and Bob agree to swap some of her tokens for some of his.

Alice and Bob do not trust one another,
nor do they both trust any third party,
so they need a way to exchange their tokens
in a safe and decentralized way.
Fortunately, they can call upon well-known
\emph{cross-chain atomic swap}
protocols~\cite{bitcoinwiki,bip199,decred,arwen2019,Herlihy2018,tiernolan,ZakharyAE2019}
that ensure that neither party can steal the other's assets.

The notion of \emph{escrow} is central to most protocols for
cross-chain exchanges.
An escrow is like a lock in a concurrent data structure:
escrowing an asset ensures that it can take part in only one exchange at a
time.
Typically an asset is escrowed by temporarily transferring ownership
to an automaton (called a ``smart contract'') programmed to award that
asset to the counterparty when certain conditions are met.
If those conditions are not met within a reasonable duration,
the escrow contract \emph{times out} and refunds the asset to its original owner.

In a typical atomic swap protocol,
Alice might put her tokens in escrow for, say, 48 hours,
then Bob would put his tokens in escrow for, say, 24 hours.
Alice would then claim Bob's tokens (details vary), and Bob would claim Alice's.
Atomic swap protocols ensure \emph{liveness}:
if both parties conform to the protocol, the swap takes place,
and even if parties deviate from the protocol,
timeouts ensure that no assets are locked up forever.
These protocols also ensure \emph{safety}:
a conforming party's assets cannot be stolen.

Nevertheless, most prior protocols have a critical flaw:
the parties are vulnerable, at different times,
to \emph{sore loser} attacks
(sometimes called \emph{lockup griefing}~\cite{arwen2019}).
Informally, this problem arises when mutually-untrusting
parties agree to a sequence of asset transfers in a volatile
market where asset values may fluctuate.
While the transfers are in progress, incentives may change:
a sudden decrease in an asset's value may motivate a party to abandon a swap midway,
or an unsuccessful bidder may drop out early from an auction.
A sore loser attack is roughly analogous to having a thread 
(deliberately) halt while holding a lock.

In the atomic swap example,
once Alice has escrowed her tokens,
Bob has the following \emph{option}
\!\footnote{In finance, this choice is called an "American call option".}:
if he observes that Alice's tokens have diminished in value,
then he stands to lose from the swap,
so he simply abandons the protocol,
leaving Alice unable to trade her tokens for 48 hours.
If, instead, Bob does respond by escrowing his own tokens,
then the balance of power reverses.
If Alice now observes that Bob's tokens have diminished in value,
then she might abandon the protocol,
leaving Bob unable to trade his tokens for 24 hours.
The sore loser attack thus introduces perverse incentives:
if either asset diminishes significantly in relative value to the other,
then one party has an incentive to quit at the other's expense.
If asset values are volatile,
parties may even have an incentive to run the protocol as slowly as possible to keep their options open for as long as possible.

The contribution of this paper is to describe novel ways
to transform various cross-chain protocols 
to mitigate or eliminate sore loser attacks.
We consider two-party atomic swaps,
multi-party atomic swaps,
brokered commerce, and simple auctions. Since the cross-chain protocols that we transform already assume synchronous communication, our mitigation mechanisms are also based on a synchronous communication model.

In classical finance,
sore loser attacks are prevented by having the option buyer
(the party who might renege)
pay a fee, called a \emph{premium},
to compensate the option seller
(the party whose assets will be locked up)
if the buyer abandons the protocol.\footnote{In some variations, the buyer pays the premium no matter what.}
Our goal here is to do the same for cross-chain commerce,
ensuring that an honest party is compensated if its assets are locked
up through no fault of its own. \footnote{The aim is to compensate honest parties, not primarily to deter or punish dishonest parties.}

Setting up a cross-blockchain premium structure presents challenges
that do not arise in classical finance.
Indeed, completely eliminating sore loser attacks seems impossible
in a distributed context,
because as soon as one party escrows an asset for the first time,
the counterparties might all renege,
leaving that party with no possibility of compensation.
If we cannot completely eliminate risk, we can still make it arbitrarily small.
We will add a premium distribution phase to protect high-value escrows from sore loser attacks.
This new phase is itself vulnerable to sore loser attacks,
but the values at risk are considerably smaller than the values
at risk in the main protocol.
For example,
Alice may be unwilling to accept a risk that 100 of her tokens could be
locked up for 48 hours,
but she may be willing to accept that risk for 1 token.

%% file: related.tex
In finance,
\emph{optionality}~\cite{higham2004introduction} is the notion that
there is value in acquiring the right,
without any obligation, to invest in something later. Atomic swap based on hashed timelock contracts \cite{tiernolan} exposes such optionality to both parties. 
However, multiple researchers~\cite{han2019optionality,arwen2019,liu2018,ZmnSCPxj} have observed that both parties are exposed to sore loser attacks where the counterparty reneges at critical points in the protocol.  Robinson~\cite{danrobinson} proposes to reduce vulnerability to sore loser attacks by splitting each swap into a sequence of very small swaps,
an approach that works only for fungible, divisible tokens.

There are prior two-party swap protocols that are \emph{asymmetric},
meaning that one party pays a premium to the other,
but not vice versa,
protecting only one side of the swap from a sore loser attack.
These protocols include Han \emph{et al.}~\cite{han2019optionality},
Eizinger \emph{et al.}~\cite{eizinger},
Liu~\cite{liu2018},
the Komodo platform~\cite{komodo},
and the Arwen protocols~\cite{arwen2019}.
Eizinger et al. \cite{eizinger} address the optionality problem by a premium mechanism,
however they address only Alice's optionality and neglect Bob's,
allowing Bob to renege after Alice escrows her assets.
Han \emph{et al.} \cite{han2019optionality} quantified optionality unfairness,
and have Alice pay premiums to Bob if she deviates from the protocol,
but not vice-versa.
In the Arwen protocols~\cite{arwen2019},
one party to each swap is a centralized exchange,
which is assumed to be trustworthy because it wants to protect its reputation.
Komodo \cite{komodo} mitigates sore loser attacks by incentives.
For example, Alice pays a small fee if she is caught deviating,
but that fee is not used to compensate Bob.

Xu \emph{et al.} \cite{xu2020game} use game-theoretic techniques to analyze the
success rate of cross-chain swaps using hashed timelock contracts,
showing that both parties can rationally choose not to follow the protocol.
To the best of our knowledge,
Liu~\cite{liu2018} was the first to propose an atomic swap protocol
that protects both parties from sore loser attacks.
This protocol is still asymmetric in the sense that Alice explicitly
purchases an option from Bob, and her premium is not refunded.
There is no obvious way to extend this protocol to more than two
parties, or to applications such as brokered sales or auctions.
Tefagh \emph{et al.}~\cite{tefaghcapital} propose a similar protocol
based on an options model.

%% file: model.tex
Although we will propose protocols based on today's blockchains and smart contracts, 
none of our principal results depends on specific
blockchain technology, or even blockchains as such.
Instead, we focus on computational abstractions central to any
systematic approach to commerce among untrusting parties,
no matter what technology underlies the shared data stores.

\subsection{Ledgers and Contracts}
A \emph{blockchain} is a tamper-proof distributed ledger (or database) that
tracks ownership of \emph{assets} by \emph{parties}.
A party can be a person, an organization, or even a contract (see below).
An asset can be a cryptocurrency, a token, an electronic deed to property, and so on.
There are multiple blockchains managing different kinds of assets.
We focus here on applications where mutually-untrusting parties trade assets among
themselves, possibly in complicated ways,
an activity sometimes called \emph{adversarial commerce}~\cite{HerlihyLS2019}.
Examples of adversarial commerce include swaps, loans, auctions, markets, and so on.

A \emph{smart contract} (or ``contract'') is a
blockchain-resident program initialized and called by the parties.
A party can publish a new contract on a blockchain,
or call a function exported by an existing contract.
Contract code and contract state are public,
so a party calling a contract knows what code will be executed.
Contract code must be deterministic because
contracts are typically re-executed multiple times by
mutually-suspicious parties.

A contract can read or write ledger entries on the blockchain where it resides,
but it cannot directly access data from the outside world,
and cannot call contracts on other blockchains.
A contract on blockchain $A$ can learn of a change to
a blockchain $B$ only if some party explicitly informs $A$ of $B$'s change, 
along with some kind of ``proof'' that the information about $B$'s state is correct.
In short,
contract code is passive, public, deterministic, and trusted,
while parties are active, autonomous, and potentially dishonest.

We assume a \emph{synchronous} execution model
where there is a known upper bound $\Delta$ on the propagation time
for one party's change to the blockchain state to be noticed by the
other parties.
Specifically, blockchains generate new blocks at a steady rate,
and valid transactions sent to the blockchain will be included in a block
and visible to participants within a known, bounded time $\Delta$. In practice, contracts measure $\Delta$ in terms of block height.

\subsection{Threat Model}\seclabel{threatmodel}
We do not consider attackers who compromise the blockchain itself,
through, for example, denial-of-service attacks.
Although parties may display Byzantine behavior,
smart contacts can enforce ordering\footnote{Smart contracts can record the order in which messages are received.}, timing, and well-formedness restrictions on
transactions that significantly limit the ways in which Byzantine
parties can misbehave.

We make standard cryptographic assumptions.
Each party has a public key and a private key,
and any party's public key is known to all.
Messages are signed so they cannot be forged,
and they include single-use labels (``nonces'')
so they cannot be replayed.

%% file: overview.tex
Cross-chain commerce is founded on the notion of \emph{escrow}.
An asset's owner does not directly transfer that asset to a counterparty.
Instead, the owner temporarily transfers that asset to an \emph{escrow contract}.
If certain conditions are met within a certain time,
that contract \emph{redeems} the asset,
transferring it to the \emph{counterparty},
and otherwise it \emph{refunds} that asset to the original owner.
The party and counterparty both trust the contract, 
even if they do not trust one another.

Multiple parties agree on a common \emph{protocol} to execute a series of transfers,
an agreement that can be monitored, but not enforced.
Instead of distinguishing between faulty and non-faulty parties,
as in classical distributed computing,
we distinguish only between \emph{compliant} parties
who follow the agreed-upon protocol,
and \emph{deviating} parties who do not.
We make no assumptions about the number of deviating parties.

Each section in this paper starts with a \emph{base protocol},
adapted from the literature,
that performs some form of adversarial commerce:
two-party swap, multi-party swaps, brokered sales, and simple auctions.
By hypothesis,
each such protocol guarantees that if all parties comply, all transfers take place,
that no asset is escrowed forever,
and that no compliant party ends up with a negative payoff
(e.g., has its assets stolen).

Nevertheless, these protocols are vulnerable to sore loser attacks:
at one or more points, one party can walk away leaving the other parties' assets
locked up in escrow for a long time.
Suppose Alice escrows an asset with value $v$
in a situation where her counterparty Bob can walk away.
A \emph{premium} is a value $p \ll v$ such that
(1) Alice considers $p$ large enough to be acceptable compensation for
locking up her asset for the duration of the protocol,
and (2) Bob considers $p$ small enough that he accepts the risk his premium could be locked up for the duration of the protocol.

We extend each base protocol to protect parties against sore loser attacks
by associating a premium with each escrow.
We modify each escrow contract to refund the premium if the asset is redeemed,
and to pay the premium to the counterparty if the asset is refunded.
The base protocol is prefaced by one or more \emph{premium distribution phases},
where premiums are deposited in escrow contracts.

These extended protocols have a semi-modular structure :
the \emph{premium protocol} vulnerable to (minor) sore loser attack
is superimposed on a \emph{base protocol} to protect it 
from (major) sore loser attack.
The premium protocol observes the base layer's state,
and as long as premium distribution completes successfully,
it does not otherwise affect the base protocol's control flow or timeouts.
If the premium protocol fails,
then so does the base protocol,
although parties may need to execute truncated versions
of the base protocol to recover their premiums.
The advantage of factoring protocols this way is that in the normal case,
when the premium protocol succeeds,
safety and liveness of the base protocol layer are independent of the premium protocol layer.

We assume each blockchain has a native currency that can be used to
pay premiums on that chain.
For simplicity, 
we treat all premiums as if they were denominated in the same currency.
For example, if Alice is pays a premium $p$ on one chain,
but receives a premium $p$ on another chain,
we say she breaks even.
For clarity when describing protocols,
we talk of \emph{depositing} premiums,
and \emph{escrowing} assets,
even though these are essentially the same mechanism.

The premiums can be estimated using formula such as the Cox-Ross-Rubinstein option pricing model \cite{cox1979option}. If the value of Alice's escrowed asset is high enough,
her minimal acceptable lock-up compensation may exceed Bob's maximum
acceptable lock-up risk, and no premium exists.
As described in \secref{bootstrap},
this mismatch can be resolved by \emph{bootstrapping} premiums:
using smaller premiums to protect distribution of larger premiums.

%% file: twoparty.tex
We now consider an atomic swap protocol where Alice and Bob exchange assets.
Like most atomic swap protocols in the literature,
ours is based on \emph{hashed timelock contracts} (HTLCs)~\cite{tiernolan}.
Alice generates a \emph{secret} $s$, its cryptographic hash $h = H(s)$,
and a \emph{timelock} $t$ after which the contract expires.
Alice publishes on the blockchain an HTLC initialized with $h,t$,
then escrows (transfers ownership of) her asset to that contract.
If the contract receives the matching secret $s$,
$h = H(s)$, before time $t$ has elapsed
\!\footnote{Since most blockchains cannot tell time directly,
$t$ is usually expressed in terms of changes to block height.},
then the contract irrevocably transfers ownership of the asset to Bob.
If the contract does not receive the matching secret
before time $t$ has elapsed,
then the asset is \emph{refunded} to Alice.
We refer to the asset being swapped as the \emph{principal}. 

\begin{figure}[!htp]
    \centering
    \includegraphics[width=.5\textwidth]{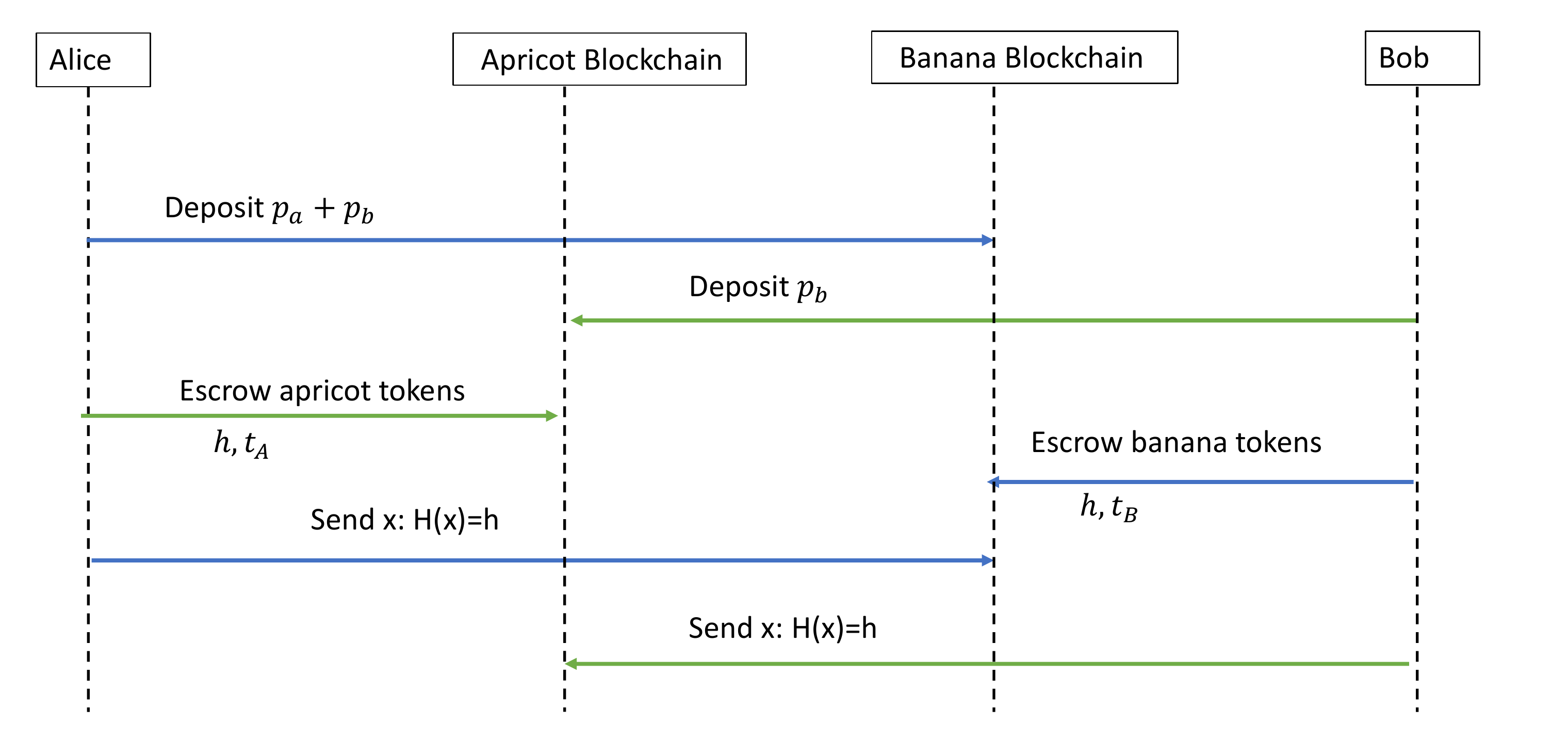}
    \caption{A Hedged Two-Party Atomic Swap Protocol}
    \figlabel{fig:onelayer}
\end{figure}

\subsection{The Base Two-Party Swap Protocol}
\seclabel{standardTwoParty}
Here is a well-known atomic swap protocol that does not protect against sore losers.
Suppose Alice wants to trade $A$ apricot tokens for one of Bob's $B$ banana
tokens.
(1) Alice generates secret $s$,
publishes an escrow contract on the apricot blockchain with hashlock $h=H(s)$
and timelock $t_A = 3\Delta$,
and escrows her apricot tokens at that contract.
(2) Within time $\Delta$,
Bob sees Alice's escrow contract on the apricot blockchain.
He publishes an escrow contract on the banana blockchain with the same hashlock $h$,
but with timelock $t_B = 2\Delta$,
and escrows his banana tokens at that contract.
(3) Within time $2\Delta$ after the start of the protocol, i.e. the time $0$,
Alice sees Bob's contract on the banana blockchain.
She sends $s$ to his contract,
acquiring Bob's principal and revealing $s$ to Bob.
(4) Within time $3\Delta$, Bob learns $s$.
He forwards $s$ to Alice's contract, acquiring Alice's principal.

This protocol guarantees that 
if both parties are compliant, the swap takes place,
that no principal is locked up forever, and that
no deviating party can steal from a compliant party.

This protocol does \emph{not} protect against sore loser attacks.
Recall that $\Delta$ is enough time for one compliant party to modify
the blockchain state (by publishing or calling a contract)
and for the other compliant party to detect that change.
To be safe, $\Delta$ should be long, say on the order of 12 hours.
If Bob walks away at Step 2, Alice's asset is locked up for $3 \Delta$,
and if Alice walks away at Step 3, Bob's asset is locked up for $\Delta$.
Bob pays no penalty for walking away.
Alice's assets remain locked up if she walks away,
but Bob gains no benefit from Alice's penalty.

\subsection{A Hedged Two-party Atomic Swap}

An atomic swap protocol should satisfy the following properties:

\begin{itemize}
    \item \emph{Liveness}. If each party is conforming, the assets are swapped and the premiums are refunded. No assets are escrowed forever.
    \item \emph{Safety}. If a compliant party transfers its asset to the counterparty, then it receives the counterparty's asset, and if it fails to receive the counterparty's asset, it does not transfer its own.
\end{itemize}

\begin{definition}
An atomic swap protocol is \emph{hedged} if, whenever a compliant party escrows assets that are not redeemed, that party receives what it considers sufficient compensation for its inability to use its escrowed assets. 
\end{definition}

Informally, by the above definition, that party's risk is limited to locking up acceptably small premiums over some bounded time, without compensation.

Suppose that Bob deposits a premium with Alice's swap contract.
What should happen to Bob's premium if he does not
unlock Alice's principal in time?
From outside, it is easy to assign blame. After Alice escrows her principal,
if Bob abandoned the swap without escrowing his principal,
then Alice is blameless, and should be awarded the premium.
If, instead, Alice abandoned the swap without revealing her secret,
then Bob is blameless, and his premium should be refunded.
Unfortunately,
the contract managing Bob's premium cannot tell the difference.
That contract resides on the apricot blockchain,
and so cannot inspect the state of Bob's escrow contract on the banana blockchain.

Here is how to solve this puzzle (See \figref{fig:onelayer}).
Say Alice's premium is $p_a$ and Bob's $p_b$.
Alice must escrow a premium of $p_a+p_b$.
If Bob reneges, his premium $p_b$ goes to Alice.
If Alice reneges, her premium $p_a+p_b$ goes to Bob,
Bob's premium $p_b$ goes to Alice,
so Bob's net compensation is $p_a$.
Alice's lock-up risk is $p_a+p_b$ until Bob's principal should be redeemed,
and Bob's risk is $p_b$ until Alice's principal should be redeemed.

For brevity,
we will often use terms like "Bob escrows his principal on the banana blockchain"
to mean
"Bob temporarily transfers ownership of his coin to an agreed-upon escrow
contract on the banana blockchain".

A contract on the apricot blockchain escrows Alice's coin and Bob's premium,
and another contract on the banana blockchain that escrows Bob's coin and Alice's premium. The timeout for the first step is $\Delta$ from the start of the protocol,
and subsequent timeouts increase by $\Delta$.

\begin{enumerate}
\item
Alice deposits her premium $p_a+p_b$ on the banana blockchain's escrow contract with timelock
$t_{A}=5\Delta $. The timeout for Alice to deposit her premium is $\Delta$. The timeout for Bob to escrow his principal $t_{b,e}=4\Delta$. 
If Bob's principal is not escrowed before $t_{b,e}$, Alice's premium is refunded. If Bob's principal is escrowed before $t_{b,e}$, then 1) if it is redeemed before $t_{A}$ elapses,
Alice's premium is refunded 2) if it is not redeemed before that timeout elapses, Alice's premium goes to Bob.
\item
Bob deposits his premium $p_b$ on the apricot blockchain's escrow contract with timelock $t_{B}=6\Delta$. The timeout for Bob to deposit his premium is $2\Delta$.
The timeout for Alice to escrow her principal $t_{a,e}=3\Delta$. 
The contract is symmetric to Bob's escrow contract. If Alice's escrowed principal is not redeemed, the premium is awarded to Alice. Otherwise, it is refunded.
\end{enumerate}
If this premium distribution phase is successful,
the parties then execute the base swap protocol,
with escrow contracts modified to transfer premiums when assets are redeemed or refunded.

It is easy to check that if Alice and Bob are both conforming,
their principals are swapped and their premiums refunded.
If Alice is the first to omit a step after Bob escrows his principal,
she will pay Bob $p_a+p_b$, and Bob will pay Alice $p_b$.
If Bob is first to deviate after Alice escrows her principal, he will pay Alice $p_b$.
Because the control flow and timeouts of the swap protocol are unaffected
by premium distribution,
the correctness of the swap phase is unaffected.

To circumvent the constraint that smart contracts on different blockchains cannot observe one another's states,
we make repeated use of the following \emph{premium passthrough} pattern.
If party $P_0$ fails to redeem escrowed asset $A_0$ on blockchain $C_0$,
then $A_0$'s escrow contract transfers premium $p$ from $P_0$ to $A_0$'s owner.
Perhaps this omission was not $P_0$'s fault because
$P_0$ was blocked by the failure of another party, $P_1$,
to redeem asset $A_1$ on a distinct blockchain $C_1$.
If $P_1$ was the source of $P_0$'s omission,
then the escrow contract for $A_1$ on $C_1$ transfers premium
$p$ from $P_1$ to $P_0$, ensuring that $P_0$ breaks even.
This passthrough pattern can be extended to sequences of arbitrary length.

%% file: bootstrap.tex
\begin{figure}[!htp]
    \centering
    \includegraphics[width=.5\textwidth]{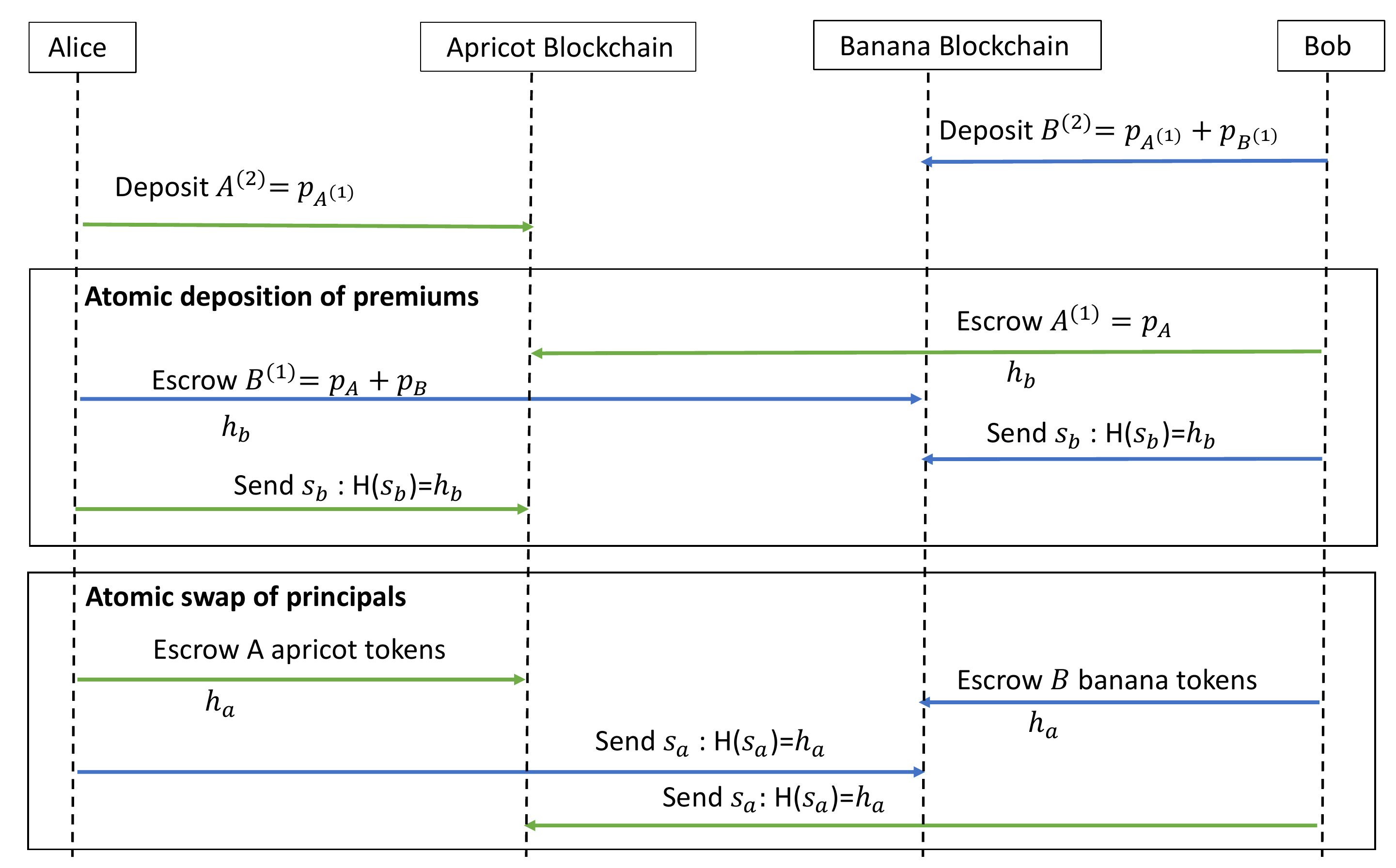}
    \caption{Hedged Two-Party Atomic Swap with 2-round Premiums}
    \figlabel{twolayer}
\end{figure}

If Alice escrows a high-value asset,
there may be no overlap between the smallest amount she will
accept as a premium and the largest amount Bob will expose to
lock-up risk.
We can reconcile this mismatch by \emph{bootstrapping} Bob's premium,
using multiple atomic swap rounds where smaller premiums are used to
protect the distribution of larger premiums.

Suppose that to protect against locking up an escrowed asset of value $v$,
Alice and Bob consider a premium  of $v/P$ for $P > 1$ to be acceptable.
In the atomic swap protocol described above,
Alice deposits a premium of value $(A+B)/P$ to Bob's escrow contract $B$,
and Bob deposits $A/P$ to  Alice's escrow contract $A$.
Alice and Bob can run a slightly modified atomic swap protocol
where instead of exchanging assets,
they deposit their premiums in their next-round escrow contracts:
Bob deposits premium $(2A+B)/P^2$ to protect the escrow of $(A+B)/P$ as Alice's next premium,
and Alice deposits $A/P^2$ to protect Bob's escrow of $A/P$.
(Since Alice's premium should be deposited first,
Bob acts as leader when they deposit premiums, see \figref{twolayer}.)
If they precede their swap with $r$ rounds of premium exchanges,
then Alice's and Bob's initial premium is $(rA+B)/P^r$ and $A/P^r$ and vice-versa, depending on who is the first leader.

Here we show a bootstrapping protocol with 2 rounds of premium deposits (\figref{twolayer}) .
We use $A^{(i)},B^{(i)}$ to denote premiums used to escrow $A^{(i-1)},B^{(i-1)}$ in the next round,
and $A^{(0)}=A$ and $B^{(0)}=B$.
In \figref{twolayer},
in the first premium deposition round,
Bob deposits $B^{(2)}$ then Alice deposits $A^{(2)}$.
In the second premium deposition round,
Alice and Bob deposit $B^{(1)}$ and $A^{(1)}$ respectively.
Bob acts as leader since he wants Alice to deposit $B^{(1)}$ to cover the next round.
Once the next round finishes,
the previous round's premiums are refunded,
except for the follower's premium in the current round.
In this example, Alice is the follower, since she will be a leader in the next round.
If she reneges, Bob has a lock-up risk of $A^{(1)}$ which exceeds Bob's acceptable risk.
Alice's $A^{(2)}$ should be refunded after Alice deposits her
principal on this escrow contract.
If Alice does not deposit her  principal,
Bob receives $A^{(2)}$ as compensation for locking up $A^{(1)}$.
Otherwise, $A^{(2)}$ is refunded and the hedged swap protocol proceeds
as in \figref{twolayer}.  

The duration of the premium lock-up risk is one atomic swap execution plus $\Delta $,
independent of the number of bootstrapping rounds.
For example, with initial premium $p$,
Alice and Bob need approximately $\log_P(\frac{A+B}{p})$ bootstrapping rounds.
With 1\% premiums and \$4 initial lock-up risk,
3 bootstrapping rounds are enough to hedge a \$1,000,000 swap.

%% file: multiparty.tex
\begin{figure}[!htp]
\begin{subfigure}[b]{.35\textwidth}
    \centering
    \includegraphics[width=.8\linewidth]{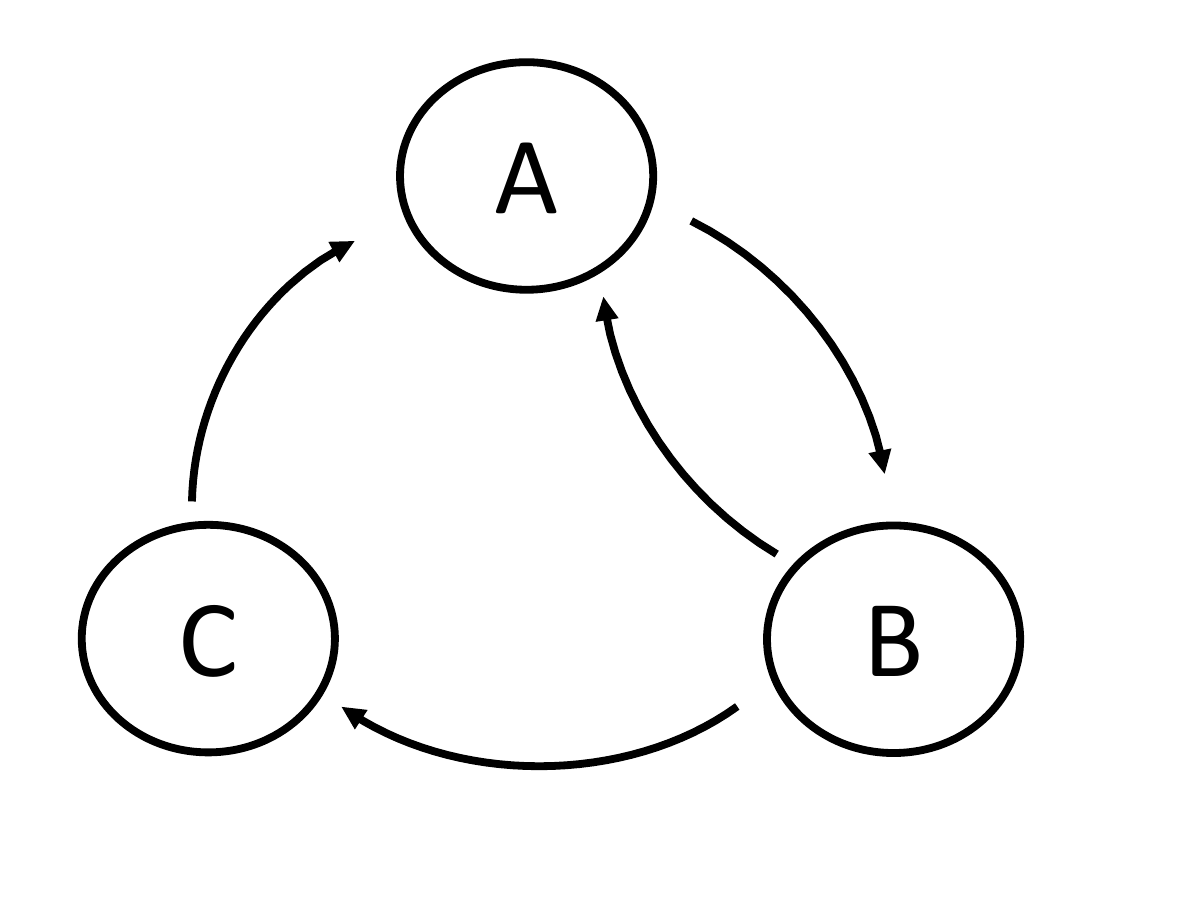}
    \caption{Multi-party Atomic Swap}
    \figlabel{multiparty}
\end{subfigure}
\begin{subfigure}[b]{.55\textwidth}
  \centering
    \includegraphics[width=.8\linewidth]{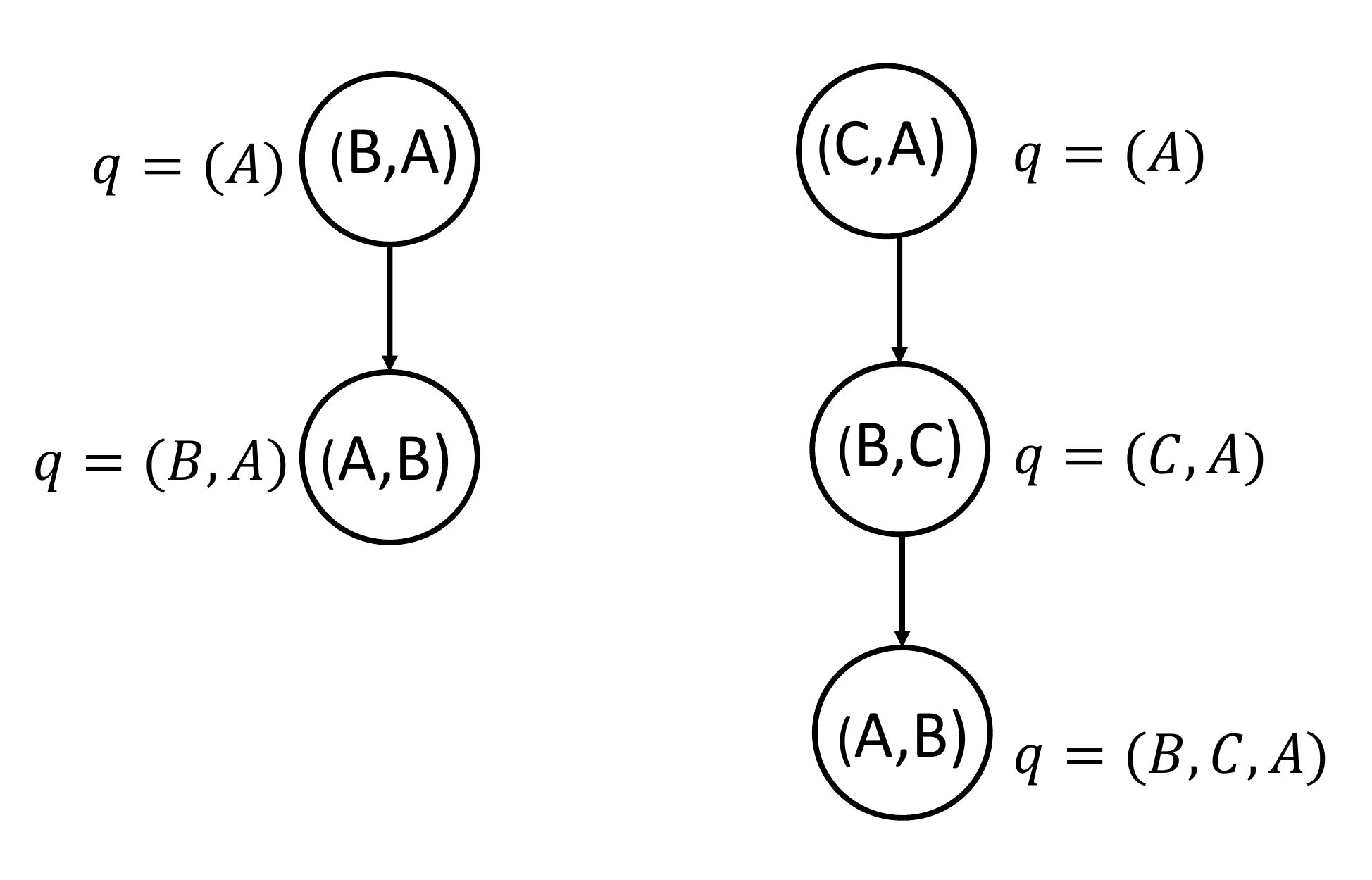}
    \caption{Paths for hashkey $k_A$}
    \figlabel{fig:pathtree}
\end{subfigure}
\caption{Multi-party Swap Digraphs}
\end{figure}

Although two-party atomic swaps are the most common in practice,
there are still situations where multiple parties want to swap assets.
A multi-party swap is represented as a strongly-connected \emph{directed graph}
(``digraph'') $\cG$ where each vertex is a party,
and each arc is a proposed asset transfer.
\figref{multiparty} shows one such swap configuration.
Henceforth, we use \emph{party} and \emph{vertex},
\emph{escrow contract} and \emph{arc}, interchangeably,
depending on whether we emphasize roles or digraph structure.

Let $\cG$ denote the swap digraph.
We say $(u,v) \in \cG$ to mean $(u,v)$ an arc of $\cG$
and similarly for vertices $v \in \cG$.
A \emph{path} $q$ in $\cG$ is a sequence of vertices
$(u_0, \ldots, u_k)$ such that each $(u_{i+1},u_i)$ is an arc of $\cG$,
and the $u_i$ are distinct.
If $u_0 = u_k$, we say $q$ is a \emph{cycle}.
\emph{Concatenation} is defined as 
$v || (u_0, \ldots, u_k) = (v,u_0, \ldots, u_k)$.

The base protocol is adapted from the multi-party swap protocol of
Herlihy~\cite{Herlihy2018}, summarized here for completeness.
See the original~\cite{Herlihy2018} for details and proofs.
Some vertices are designated as leaders, the rest as followers.
The leaders must form a feedback vertex set \footnote{A feedback vertex set is a subset of vertices whose deletion leaves
$\cG$ acyclic.} in the digraph. 
Each leader $L_i$, for $i \in 1..\ell$,
generates a secret $s_i$ and hashlock value $h_i = H(s_i)$,
yielding a \emph{hashlock vector} $(h_1, \ldots, h_\ell)$,
which is sent to each arc.
A \emph{hashkey} $k_i$ for $h_i$ on arc $(u,v)$ is a triple $(s_i,q,\sigma)$,
where $s_i$ is the secret $h_i = H(s_i)$,
$q$ is a path $(u_0,\ldots,u_i)$ in $\cG$
where $u_0 = v$ and $u_i$ is the leader who generated $s_i$,
and $\sigma$ is a sequence of signatures that authenticates the path $\sigma = sig(\cdots sig(s_i,u_i),\cdots,u_0)$. \figref{fig:pathtree} shows how paths are collected on each arc in $\cG$ of \figref{multiparty} where Alice is the only leader generating $s_a$. The nodes represents arcs.
A hashkey $(s_i, q, \sigma)$ \emph{times out} at time
$(\diam(\cG) + |q|) \cdot \Delta$ after the start of the protocol.
A hashkey no longer unlocks its hashlock after it times out.
That hashkey $(s_i, q, \sigma)$ \emph{unlocks} the hashlock $h_i$ on
$(u,v)$ if it is presented before it times out.

The base protocol has two phases.
In Phase One, each leader
(1) escrows an asset on every arc leaving that vertex, then
(2) waits until assets have been escrowed on all arcs entering that vertex.
Each follower
(1) waits until assets have been escrowed on all arcs
entering that vertex, then
(2) escrows an asset on every arc leaving that vertex.
In Phase Two, 
each leader whose incoming arcs have the expected escrowed assets
sends its hashkey to those arcs.
Each party who learns a hashkey from an incoming arc
extends that hashkey's path and propagates the extended hashkey on its outgoing arcs.
When an arc has collected all hashkeys needed,
the asset escrowed in that arc is \emph{redeemed} and transferred to the counterparty.

The base multi-party swap protocol satisfies the same safety properties as the two-party swap:
for each compliant party $v$,
(1) if $v$ transfers an asset on an outgoing arc,
then it receives all assets on incoming arcs,
and (2) if $v$ fails to receive an asset on an incoming arc,
then it transfers no assets on any outgoing arcs.

\subsection{Premium Distribution}
The two-party premium distribution protocol of \secref{twoparty}
does not easily generalize to multi-party swaps.
Consider the graph in \figref{multiparty}.
Suppose Alice posts premiums on her incoming edges $(B,A)$ and $(C,A)$.
In Phase One, (conforming) Bob escrows his assets on $(B,A)$,
but (deviating) Carol never escrows hers.
Alice has a dilemma.
If she releases her secret, Bob will take her asset,
but she will not get Carol's asset in return.
If she does not release her secret,
she will have to pay a premium to Bob.
The dilemma arises because Alice's counterparty in a two-party swap is
either compliant or deviating, 
but in a multi-party swap, her counterparties may include both.

There are two ways a deviating party can lock up its counterparties' assets.
In Phase One, a deviating party may fail to escrow its principal,
and in Phase Two, 
it may fail to deliver a hashkey needed to redeem an asset.
In response,
we define two kinds of premiums for each arc $(u,v)$:
an \emph{escrow premium} is awarded to $v$ by $u$ if the
expected asset is not escrowed on that arc in time,
and 
a \emph{redemption premium} is awarded to $u$ by $v$
if $v$ does not produce the hashkey $k_i$ in time.

Premiums are deposited in two phases that
mirror the phases of the base protocol:
first the escrow premiums are deposited, then the redemption premiums.
It is convenient to describe these protocols in reverse chronological order, redemption premiums first.
Redemption premiums flow ``backwards'' though the digraph,
starting at leaders, and moving against the orientation of the arcs.
Consider hashkey $k_i$ from leader $L_i$.
A redemption premium for arc $(u,v)$ has the form $R_i(q,u)$,
where $q$ is a path from $v$ to $L_i$ in $\cG$.
(This path reverses the order in which that premium was distributed.)
Exactly as in Phase Two of the base protocol,
this path is authenticated by signatures,
and the path length determines timeouts.

Here is the redemption premium distribution protocol for leader $L_i$.
Protocols for different leaders can be run in parallel.
Assume for simplicity that each asset has the same premium $p$.
\begin{enumerate}
\item $L_i$ deposits premium $R_i(L_i, u)$ on each incoming arc $(u,L_i)$, and
\item waits until each outgoing arc $(L_i,v)$ has a premium for $k_i$.
\end{enumerate}
Each party $v \neq L_i$,
\begin{enumerate}
\item
  waits for the first time a premium $R_i(q,v)$ for $k_i$ appears on
  some outgoing arc $(v,w)$, then
\item if $v||q$ is a path,
  then deposits premium $R_i(v||q,u)$ on every incoming arc $(u,v)$.
\end{enumerate}
Once a premium for $k_i$ has appeared on any of $u$'s outgoing arcs,
any later premiums for $k_i$ that appear on other outgoing arcs are
ignored.
The proof that the redemption premium distribution protocol
terminates is identical to the proof that the hashkey distribution
phase of the base protocol terminates,
which appears elsewhere~\cite{Herlihy2018}.
If this phase times out, the party still goes to next phase. 

How are redemption premiums calculated?
Each party $v$'s redemption premium for path $q$ is:
\begin{equation}
\eqnlabel{redeem-formula}
R_i(q,v) =
\begin{cases}
  p &\text{if } v||q \text{ is a cycle} \\
  p + \sum_{\set{u | (u,v) \in \cG}}R_i(v||q,u) &\text{otherwise}
\end{cases}
\end{equation}
This formula is well-defined because each path in $\cG$ is finite, being acyclic.
Each leader's redemption premium is
\begin{equation*}
R(L_i) = \sum_{\set{u |(u,L_i) \in \cG}} R_i(L_i,u),
\end{equation*}
the sum of the premiums on its incoming arcs.

Escrow premiums propagate ``forwards'' through the digraph,
passing from asset sender to asset receiver. Let $E(u,v)$ denote the escrow premium on arc $(u,v)$.
Each leader $L$
\begin{enumerate}
\item 
  deposits premium $E(L,v)$ on each outgoing arc $(L,v)$, and
\item
  waits until premium $E(u,L)$ has been deposited on each incoming
  arc $(u,L)$.
\end{enumerate}
Each follower $F$
\begin{enumerate}
\item 
  waits until premium $E(u,F)$ has been deposited on each incoming arc $(u,F)$, and
\item
  deposits premium $E(F,v)$ on each outgoing arc $(F,v)$.
\end{enumerate}
The proof that the escrow premium distribution protocol
terminates is identical to the proof that the escrow
phase of the base protocol terminates,
which appears elsewhere~\cite{Herlihy2018}.
If this phase times out,
the party still moves to the next phase of the protocol.  

Before an escrow premium deposited by $u$ can be awarded to $v$,
that premium must be \emph{activated}.
A premium deposited on arc $(u,v)$ is activated when
$(u,v)$ has received redemption premiums for all hashkeys $k_i$.
If $u$'s escrow premium times out before activation,
it is refunded to $u$, but after activation,
it is awarded to $v$ if the asset on $(u,v)$ is not escrowed in time.

Escrow premiums are computed by the following formula:
\begin{equation}
\eqnlabel{escrow-formula}
E(u,v) =
\begin{cases}
  R(L_i) &\text{ if } v \text{ is leader } L_i \\
  \sum_{(v,w)\in \cG}E(v,w) &\text{otherwise.}
\end{cases}
\end{equation}
The first clause states that each arc entering a leader
carries a premium equal to that leader's redemption premium.
The second clause states that each arc entering a follower
covers the premiums on arcs leaving that follower. 
The escrow premium formula is well-defined because leaders form a
feedback vertex set, so every cycle is broken by a leader.

The hedged protocol has four phases:
(1) depositing escrow premiums,
(2) depositing redemption premiums,
(3) base protocol Phase One, and
(4) base protocol Phase Two.
If the first two premium distribution phases execute successfully,
the base protocol phases execute normally,
with some additional steps to manage premiums.
If premium distribution fails,
the parties execute truncated versions of the base protocol phases
to recover their premiums.

Timeouts are determined as follows. 
Each step takes time at most $\Delta$. 
In the first phase, the leaders should escrow their outgoing escrow premiums before $\Delta$ elapses,
and each following step's timeout increases by $\Delta$.
Premiums and assets are locked until they are due to be activated, redeemed, or refunded.

In the following lemmas, $v$ is a compliant party.
\begin{lemma}
If a swap completes successfully,
then each $v$ has all its premiums refunded.
\end{lemma}

\begin{proof}
For each outgoing arc $(v,w)$,
$v$'s escrow premium $E(v,w)$ is refunded as soon as $v$ escrows its
asset on that arc.
For each incoming arc $(u,v)$ and each hashkey $k_i$,
$v$'s redemption premium $R_i(q,u)$ is refunded as soon as $v$
party sends hashkey $k_i$ on that arc.
\end{proof}

What can go wrong?
\begin{lemma}
\lemmalabel{phase-four}
In Phase Four,
if hashkey $k_i$ is never revealed on any of $v$'s outgoing arcs,
then $v$ ends up with net redemption premium profit at least $p$
for each asset $v$ escrowed.
\end{lemma}

\begin{proof}
Denote the redemption premium that $v$ receives on each outgoing arc $(v,w)$ as $R_i(q,v)$, where $q=(v,w,\cdots,L_i)$. If $v$'s outgoing redemption premiums on incoming arcs have the form $R_i(v||q,u)$ for all arc $(u,v)\in \cG$,
then by \eqnref{redeem-formula},
$v$ ends up with net redemption premium profit at least $p$ for any outgoing arc $(v,w)$.
If $v||q$ is a cycle,
$v$ is awarded at least $p$ for any outgoing arc $(v,w)$. Every outgoing arc contributes a redemption premium profit at least $p$.
\end{proof}

\begin{lemma}
\lemmalabel{phase-three}
In Phase Three,
if some $u$ fails to escrow an asset on $(u,v)$,
then a party $v$ ends up with a net escrow profit at least zero,
and a net redemption premium profit at least $p$
for each asset $v$ escrowed.
\end{lemma}

\begin{proof}
If $u$ fails to escrow its asset,
then $v$ collects the escrow premium on $(u,v)$.
By \eqnref{escrow-formula},  if $v$ is a follower, $v$ does not escrow any outgoing assets, and the premium $E(u,v)$ is enough to cover the cost of paying the escrow premiums on $v$'s outgoing arcs,  yielding the net escrow premium profit at least zero.
If $v$ is a leader, since the leader $v$ escrows assets on the outgoing arcs whose escrow premiums are activated, $v$ does not need to pay escrow premiums to anyone, yielding the net escrow premium profit at least zero. The leader $v$ then proceeds to Phase Four without revealing $k_v$, the premium $E(u,v)$ is enough to cover the cost of paying the redemption premiums on $v$'s incoming arcs. For any arc $(v,w)$ that $v$ has escrowed asset, since the escrow premium $E(v,w)$ is activated,  and $k_v$ on $(v,w)$ cannot be revealed, yielding per-asset net redemption premium profit at least $p$ by \lemmaref{phase-four}.
\end{proof}

\begin{lemma}
\lemmalabel{phase-two}
In Phase Two,
if no redemption premium for $k_i$ is deposited on any outgoing arc $(v,w)$,
then $v$ ends up with a net escrow premium profit at least zero,
and a net redemption premium profit at least zero.
\end{lemma}

\begin{proof}
If no redemption premium for $k_i$ is deposited on any outgoing arc $(v,w)$,
then $v$'s escrow premiums on those arcs are not activated,
and they are all refunded to $v$,
for a net escrow premium profit at least zero. If $v \neq L_i$,
then $v$ does not deposit any redemption premiums for $k_i$ on any incoming arc,
for a net $k_i$ redemption premium profit at least zero. If $v=L_i$, $v$ proceeds to Phase Three  without escrowing any assets since no escrow premium is activated. $v$ just releases $k_i$ on its incoming arcs and 
gets a net $k_i$ redemption premium profit at least zero since $v$ does not pay any premium.
\end{proof}

\begin{lemma}
\lemmalabel{phase-one}
In Phase One,
if some $u$ fails to deposit an escrow premium on $(u,v)$,
then $v$ ends up with a net escrow premium profit at least zero.
\end{lemma}

\begin{proof}
If $v$ is a follower, $v$ does not deposit any escrow premiums since $v$ does not receive all incoming escrow premiums. If $v$ is a leader, $v$ proceeds to Phase Two without depositing the redemption premium $R_v(v,w)$ for any $(w,v)\in \cG$ and its outgoing escrow premiums are refunded eventually.
\end{proof}

\begin{lemma}
The multi-party swap protocol is hedged.
\end{lemma}

\begin{proof}
\lemmarefrange{phase-four}{phase-one} imply that
in every situation where a compliant party escrows an asset,
it ends up with a premium profit of at least $p$ for that asset.
\end{proof}

A leader deposits a premium proportional to the number of paths in the digraph.
If there is a unique path between any two parties,
then each leader's premium is linear in $n$,
the number of digraph vertices.
In the worst case, for a complete digraph,
each leader's premium is exponential in $n$.
This premium can be reduced to linear by preceding the protocol with
$O(\log n)$ rounds of premium bootstrapping as described in \secref{bootstrap}.

%% file: broker.tex
Not all cross-chain commerce can be expressed as swaps.
Consider the following scenario from Herlihy, Liskov, and Shrira~\cite{HerlihyLS2019}.
Alice is a ticket broker who buys tickets at wholesale prices from event organizers
and resells them at a small markup to consumers.
Alice discovers that Bob wants to sell some tickets for 100 coins,
and Carol is willing to buy them for 101 coins,
so Alice wants to broker the deal.
This three-way exchange is not a swap,
because Alice does not own either the tickets or the coins:
she is using Carol's coins to buy Bob's tickets.
Coins and tickets live on distinct blockchains.

\subsection{Base Protocol}
\begin{figure}[!htp]
\begin{subfigure}[b]{.35\textwidth}
    \centering
    \includegraphics[width=.8\linewidth]{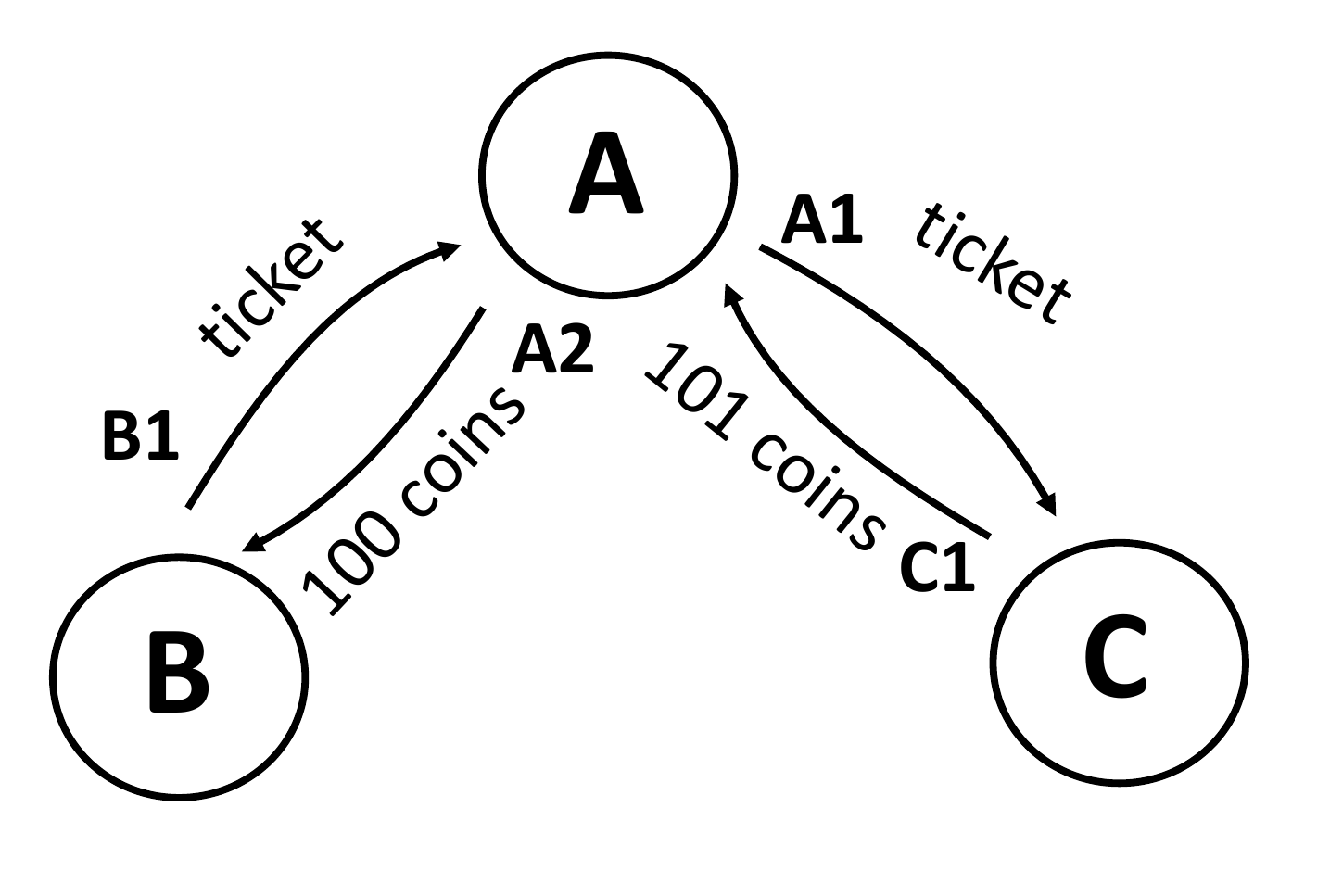}
    \caption{Broker Protocol Digraph}
\end{subfigure}
\begin{subfigure}[b]{.5\textwidth}
    \centering
    \includegraphics[width=\linewidth]{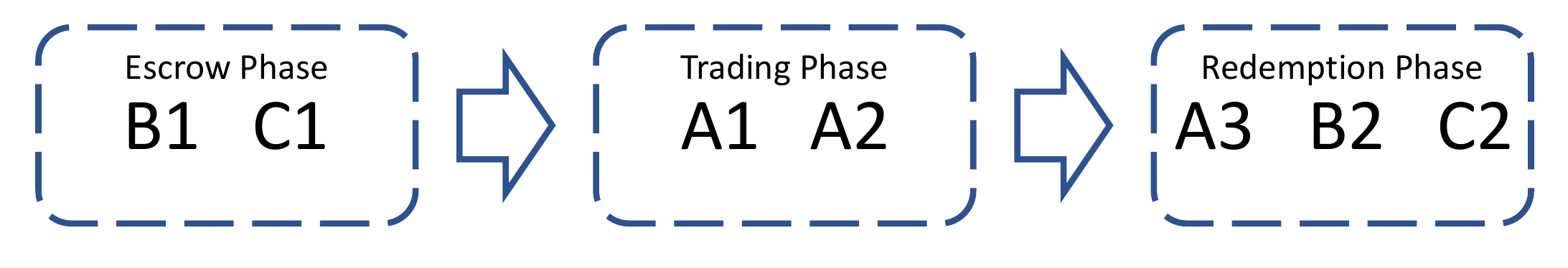}
    \caption{Broker Protocol Phases}
\end{subfigure}
    \caption{Broker Protocol}
    \figlabel{broker}
\end{figure}
The broker protocol summarized here is taken from
Herlihy, Liskov, and Shrira~\cite{HerlihyLS2019},
which includes a more complete analysis.
In the terminology of the multi-party swap,
every party is a leader.
Each party $X \in \set{A,B,C}$,
generates a secret $s_X$ and hashlock value $h_X = H(s_X)$,
yielding a \emph{hashlock vector} $(h_A, h_B, h_C)$,
which is sent to each arc.
A \emph{hashkey} $k_X$ for $h_X$ on arc $(u,v)$ is a triple $(s_X,q,\sigma)$,
where $s_X$ is the secret $h_X = H(s_X)$,
$q$ is a path $(u_0,\ldots,u_k)$ in $\cG$
where $u_0 = v$ and $u_k$ is the party who generated $s_X$,
and $\sigma$ is a sequence of signatures that authenticates the path.
An asset is redeemed when its arc has received all three hashkeys in time.
As in the multi-party swap protocol,
a hashkey $(s_i, q, \sigma)$ \emph{times out} at time
$(\diam(\cG) + |q|) \cdot \Delta$ after the start of the protocol.

Here are the steps of the base protocol.
\begin{enumerate}
  \item Escrow phase:
    \textbf{B1}: Bob escrows his tickets on arc $(B,A)$, and
    \textbf{C1}: Carol escrows 101 coins on arc $(C,A)$.
  \item Trading phase:
    \textbf{A1}: Alice transfers the tickets to Carol on $(A,C)$,
    \textbf{A2}: Alice transfers 100 coins to Bob on $(A,B)$.
  \item Redemption phase:
    \textbf{A3}: Alice releases her hashkey on $(B,A),(C,A)$
    \textbf{B2}: Bob releases his hashkey on $(A,B)$, and
    \textbf{C2}: Carol releases her hashkey on $(A,C)$.
    When a party observes a new hashkey on an outgoing arc,
    it propagates that hashkey to all its incoming arcs.
\end{enumerate}

\subsection{Premium Structure}
Who should pay premiums to whom?
\figref{broker} shows the dependencies among these steps.
An arrow from one step, say B1, to another, A1, means that A1
cannot occur until B1 has happened.
If Bob omits B1, then Carol's coins are locked up,
and Alice cannot complete A1, forcing her to pay a premium,
so Bob pays a premium to Carol and to Alice.
If Bob completes B1 but omits B2,
then Carol's coins are locked up,
so he pays a premium to Carol.
Carol's premium payments are symmetric.

Alice's situation is the most interesting,
since her role would not exist in a multi-party swap.
She escrows no assets,
but she should still receive passthrough premiums to reimburse her for
premium payments forced on her by others.
If Alice omits A1 after Bob performs B1,
then she pays Carol a premium on the ticket blockchain since Carol is expecting her to transfer the ticket to her.
If she omits A2 after Carol performs C1,
then Alice pays Bob a premium on the coin blockchain.
If she omits A3 after Bob and Carol complete B1, B2, C1, and C2,
then she pays premiums to both on their respective blockchains.

Premiums are deposited in a three-phase protocol mirroring
the structure of the base protocol.
\begin{enumerate}
\item In the escrow premium deposit phase,
  Bob and Carol, the parties escrowing their own assets,
  deposit escrow premiums $E(B,A)$ and $E(C,A)$ on those outgoing arcs.
\item In the trading premium deposit phase,
  Alice, the only party conducting intermediate trades,
  deposits trading premiums $T(A,B)$ and $T(A,C)$ on those outgoing arcs.
\item In the redemption premium deposit phase,
  for each $u \in \set{A,B,C}$,
  each $v$ deposits premium $R_u(q,u)$ on each incoming arc $(u,v)$,
  where $q$ is a path from $v$ to $L$\footnote{In this specific case, there are opportunities for optimization. Since $(A,B)$ and $(C,A)$ are asset transfers on the same escrow contract, Bob can directly send his hashkey to the coin blockchain, simplifying the redemption premium deposition.  Since Alice does not need to forward Bob's hashkeys on $(C,A)$ , we do not need a redemption premium regarding the path $q=(A,B)$ for $s_b$. The ticket chain is symmetric. } .
\end{enumerate}
As in the multi-party swap protocol,
an escrow or trading premium is \emph{activated} on an arc
when all redemption premiums have been deposited on that arc.
As long as an escrow or trading premium has not been activated,
it can only be refunded.

Redemption premiums are calculated by \eqnref{redeem-formula}.
Trading premiums are defined as follows:
if $v$ transfers an asset to $w$ in the trading phase,
then $v$'s trading premium $T(v,w)$ is $R_w(w)$.
Escrow premiums are similar:
Let $T(v) = \sum_{w | (v,w) \in \cG} T(v,w)$.
then $u$'s escrow premium $E(u,v)$ is $T(v)$.

As long as all trading-phase transfers are known in advance,
we can extend this approach to encompass multiple rounds of trading.
Premiums for $r$ trading rounds are defined as follows.
If $v$ transfers an asset to $w$ in the escrow phase,
then $v$'s escrow-phase premium $E(v,w)$ is $T_1(w)$.
If $v$ transfers an asset to $w$ in trading phase $k$, $1 \leq k < r$,
then $v$'s phase-$k$ trading premium $T_k(v,w)$ is $T_{k+1}(w)$.
If $v$ transfers an asset to $w$ in trading phase $r$,
then $v$'s phase-$r$ trading premium $T_{r}(v,w)$ is $R_w(w)$.
In an $r$-round deal, assets change hands $r$ times.

%% file: auction.tex
Consider a scenario where Alice has purchased some tickets she now
wants to auction to Bob and Carol.
What happens if we na\"ively try to adapt \secref{broker}'s broker protocol?
If Bob submits the higher bid,
but Alice dishonestly tries to take his coins without awarding him the tickets,
then Bob simply cancels the auction by withholding his final vote to
commit (his hashkey).
Bob is safe, but he is exposed to a sore loser attack:
if Carol is angry because her bid lost,
she withholds her vote to commit,
ensuring that no one gets the tickets.
A premium structure similar to the hedged broker protocol
could compensate Bob if sore loser Carol wrecks the auction,
but suppose again that a dishonest Alice tries to take Bob's money without
awarding him the tickets.
When Bob justifiably withholds his vote to commit,
he will be unfairly required to pay premiums to the others.

In this section, we propose a simple hedged auction protocol
that is not vulnerable to a sore loser attack from the low bidder,
and that compensates the bidders if the auctioneer is caught cheating.

\subsection{Base Protocol}
As \secref{broker},
there are two blockchains, the ticket chain and the coin chain.
Alice generates two secrets:
$s_B$ to be used if Bob wins, and $s_C$ if Carol wins.
Alice constructs hashkeys $k_B$ based on $s_B$, and $k_C$ based on $s_C$.
Recall that a hashkey is a triple $(s,q,\sigma)$,
where $q$ is the path the hashkey has traversed,
$s$ is a secret, and $\sigma$ the signatures authenticating the path.
The hashkey times out after time $|q|\Delta$.
Since there are only 3 parties,
the longest a hashkey can survive is $3\Delta$.
For brevity, we use $k_B$ ($k_C$) to denote any valid hashkey based on
$s_B$ ($s_C$).
The protocol has several phases, each of duration $\Delta$.
\begin{enumerate}
\item
  In the \emph{bidding phase},
  Bob and Carol send their bids\footnote{
    In a more realistic auction protocol, the bidders might use a two-round
    commit-reveal scheme to keep their bids secret from one another,
    a topic beyond this paper's scope.}
  to the coin chain contract, which records them.
  At the end of this phase, 
  the high bidder's identity is evident from inspecting the coin chain
  contract.
  No new bids are accepted after this phase.

\item
  \label{declaration-phase}
  In the \emph{declaration phase},
  Alice inspects the coin chain contract to determine the winner,
  and publishes the hashkey identifying the winner on both the coin
  and ticket chain contracts.
  (For example, if Bob wins, she publishes $k_B$.)

\item
  \label{challenge-phase}
  In the \emph{challenge phase},
  Bob and Carol inspect the hashkeys Alice published on the coin and
  ticket chain contracts, if any.
  If any hashkey appears at one contract but not the other,
  Bob and Carol forward that hashkey to the contract missing that hashkey.
  This phase takes time $3\Delta$,
  long enough for Alice's hashkeys to time out.
\item
  \label{commit-phase}
  In the \emph{commit phase}, the auction is settled.
  The coin chain contract compares the hashkeys it has received with the bids.
  If it received only the actual winner's hashkey, all is well, and
  it refunds the lower bid and transfers the higher bid to Alice.
  If it received the low bidder's hashkey, or no hashkey,
  then Alice cheated, and all bids are refunded.
  If the ticket chain contract received exactly one hashkey,
  it transfers the tickets to the matching party.
  If it received zero or two hashkeys, it refunds the tickets to Alice.
\end{enumerate}

\begin{lemma}
If a hashkey $k=(q,s,\sigma)$ is published on one contract,
then it is also published on the other.
\end{lemma}
\begin{proof}
  If the path $q$ includes a compliant party,
  then that party has already published $k$ on the other contract.
  If path $q$ does not include any compliant party,
  then $q$ has length at most 2,
  implying $k$ was published before $2\Delta$ elapsed.
  The missing compliant party has time $\Delta$ to publish $k$
  on the other contract before $k$ times out.
\end{proof}  

\begin{lemma}
  No compliant bidder's bid can be stolen.
\end{lemma}

\begin{proof}
  Suppose Bob is the high bidder. 
  
  If no hashkeys are published on either contract,
  then all bids are refunded at Phase~\ref{commit-phase}. 
  
  If any party publishes $k_C$ on either contract,
  then some party publishes $k_C$ on the coin contract,
  and all bids are refunded at Phase~\ref{commit-phase}. 
  
  If no party publishes $k_C$ on either contract,
  but some party publishes $k_B$ on some contract,
  then $k_B$ and only $k_B$ is published on both contracts,
  so the chain contract will refund Carol's bid and transfer Bob's bid to Alice,
  and the ticket blockchain will transfer the tickets to Bob.
\end{proof}

If Alice deviates, she can award the tickets to either bidder (or neither),
but since she owns those tickets, she could have done that without an auction.
What matters is that if Bob or Carol are compliant, their bids cannot be stolen.

\subsection{Premium Structure}
Bob and Carol do not pay premiums because they cannot maliciously lock
up anyone's assets.
(A party who withholds a bid arguably does the other party a favor.)
Alice should pay premiums,
because she can lock up Bob and Carol's coins,
either by abandoning the protocol midway or by cheating.

Alice endows her coin chain contract with $2p$ premiums.
If the bids are refunded in Phase~\ref{commit-phase},
then Bob and Carol are each awarded $p$ along with their refunded bids.
If the auction completes, Alice's premiums are refunded.
Generalizing this protocol to $n$ bidders
requires Alice to deposit $np$ premiums.

%% file: conclusions.tex
We used model checking to verify the properties of the two-party
hedged swap and some three-party hedged swaps.
As discussed in \secref{threatmodel},
smart contracts severely constrain the behavior of Byzantine participants
by enforcing ordering, timing, and well-formedness restrictions on
transactions.
Byzantine parties are restricted to attacks that appear reasonable
at individual blockchains, even if they are globally incorrect.
Surprisingly, perhaps, 
this constrained behavior can be model-checked in reasonable time.
The TLA+ source code and model specifications can be found in our GitHub repository \cite{githubrepo}.

Blockchains such as Ethereum, whose smart contracts are implemented using a Turing-complete language, can support our protocols directly. For blockchains such as Bitcoin, whose contracts are more restricted, we note that Han \emph{et al.} \cite{han2019optionality} introduced a new opcode to support their premium protocol. 

 In future work, we plan to study premiums in asynchronous protocols such as those proposed by Glabbeek \emph{et al.}~\cite{van2020feasibility}, Ranchal-Pedrosa and Gramoli~\cite{ranchal2019platypus} and Herlihy \emph{et al.}~\cite{HerlihyLS2019}. 
 
 We have made no attempt to optimize the round complexity of our protocols. It would be interesting to derive lower bounds for round complexity of premium protocols.

We have studied the sore loser problem in the context of cross-blockchain financial deals, but similar issues arise in any distributed coordination protocol where a faithless party can trick another party into locking up a resource for a non-trivial duration.
The resources at risk might be disk pages, network bandwidth, database access,
and so on.
In everyday life,
there are well-developed mechanisms for sore loser protection,
such as security deposits, "earnest money", downpayments and so on,
and we hope this paper will focus the community's attention
on developing similar mechanisms for distributed computing.